\documentclass[twoside,leqno,twocolumn]{article}  
\usepackage{ltexpprt}

\newtheorem{thm}{Theorem}

\usepackage{graphicx}
\usepackage{epstopdf}
\usepackage{epsfig}

\usepackage{hyperref} 

\usepackage{times}
\usepackage[square,numbers]{natbib}
\usepackage{courier}
\usepackage{amssymb}
\usepackage{amsmath}

\begin{document}

\title{\Large MaxRank: Discovering and Leveraging the Most Valuable Links for Ranking
}
\author{Hengshuai Yao \\
University of Alberta\\
hengshua@cs.ualberta.ca
\and 
}
\date{}

\maketitle

\begin{abstract}
On the Web, visits of a page are often introduced by one or more valuable linking sources.
Indeed, good back links are valuable resources for Web pages and sites.
We propose to discovering and leveraging the best backlinks of pages for ranking. 
Similar to PageRank, MaxRank scores are updated {recursively}. 
In particular, 
with probability $\lambda$, 
the MaxRank of a document is updated from the backlink source with the maximum score;
with probability $1-\lambda$, 
the MaxRank of a document is updated from a random backlink source.
MaxRank has an interesting relation to PageRank.
When $\lambda=0$, MaxRank reduces to PageRank; 
when $\lambda=1$, MaxRank only looks at the best backlink it thinks.
Empirical results on Wikipedia shows that 
the global authorities are very influential; 
Overall large $\lambda$s (but smaller than $1$) perform best: 
the convergence is dramatically faster than PageRank, but the performance is still comparable.
We study the influence of these sources and propose a few measures such as the times of being the best backlink for others, and related properties of the proposed algorithm. 
The introduction of best backlink sources provides new insights for link analysis. 
Besides ranking, 
our method can be used to discover the most valuable linking sources for a page or Website, 
which is useful for both search engines and site owners.

\end{abstract}

\section{Introduction}\label{sec:intro}
The gigantic size and diverse content of modern databases have made ranking algorithms fundamental components of search systems \citep{searchweb}. 
The link analysis approach to ranking has been proven to be very effective in evaluating the qualities of Webpages \citep{HITS,Page98}, 
with widely practice from industry and intensive studies from academics. 
The success has proven that the hyperlinks on the Web are useful in finding high quality sources, which is hard based only on the content of pages.
PageRank and HITS are two seminal algorithms in literature. 
PageRank finds {\em authorities} which are the pages frequently visited by a random surfer. 
HITS finds both {authorities} and {hubs}, which are defined recursively---
the {\em authorities} are frequently linked by the {\em hubs} which turn out to be the pages frequently linked by authorities.
In this paper, we will be focused on finding authorities in the spirit of PageRank, though our techniques may also apply to HITS and other link analysis algorithms. 

{\em PageRank}.
In PageRank formulation, with probability $c$, 
a random surfer model follows the links on a page uniformly at random, 
and with probability $1-c$, the surfer model jumps to a new page selected uniformly at random from the database. 
The PageRank value of a page is defined as the probability of visiting the page in the long run of the random walk, e.g., see \citep{Page98,PR_survey,deeper04,pr_survey3,amy06,liu07}.

Suppose there are $N$ documents  in the database. 
All vectors are column vectors. The transpose of a matrix $X$ is denoted by $X^T$.
We need the following notations. 

$L$ be an adjacency matrix of the database.
That is, $L(i,j)=1$ if there is a link from document $i$ to document $j$, otherwise $L(i,j)=0$, $i,j=1,2, \ldots, N$;

$\bar{L}$ be a row normalized matrix of $L$;

$e$ be a vector of all $1$s, and $v$ be a vector of probabilities that sum to one; and

$S$ be a stochastic matrix such that $S=\bar{L}+ (a  e^T/N)$, 
where $a_i=1$ if document $i$ is dangling (i.e., document $i$ has no forward link)
and $0$ otherwise.

The transition probability matrix used by PageRank is 
\[
G=c S+ (1-c) ev^T,
\]
where $v$ (often called the {\em teleportation vector}) is a probability vector that sums to one.
Matrix $G$ is sometimes called the {\em Google matrix} in literature \citep{amy06}. 
One merit of the Google matrix is that it is stochastic and primitive and thus its steady state distribution (also called the stationary distribution) exists. 
In fact, 
PageRank (denoted by $\pi$) is exactly the steady state distribution vector of $G$, satisfying 
\[
{\pi}=G^T{\pi}.
\]
The other merit of $G$ is that it does not have to be stored, 
and the power iteration of computing ${\pi}$ can take advantage of the rank-1 matrix $ev^T$, 
manipulating $S$, $c$, $e$ and $v$ directly, e.g., see \citep{Taherthesis}.

Considerable efforts have been devoted to the computation problem of PageRank due to its large scale applications. 
This is especially important when one wants to compute multiple PageRank vectors depending on queries and users. 
For a detailed discussion, please refer to Section \ref{sec:discussion}.
In this paper,
we present a method utilizing the best backlinks, which have a much faster convergence than PageRank but the performance is still comparable. 
We are interested in understanding the roles of these influential links and their implication for link analysis algorithms especially PageRank. 

\section{Research Questions}\label{sec:questions}
Link analysis takes advantage of the linking information in calculating document importances. 
For example, Pagerank uses the back links of a document in updating its score. 
Intuitively, there are {\em influential} links which contribute a large portion to the score, 
and there are unimportant links which only contribute a negligible portion.  
We would like to ask the following questions. 

\begin{itemize}

\item 
Where are the influential links from? 
What types of documents are the influential sources? 
Are they authorities, hubs, or anything else?
What relations are they to the nodes that are influenced by them?

\item 
How many such influential sources are there?

\item 
How influential is a backlink to the score of a document?
Most importantly, how influential are those influential back links?

\end{itemize}

These questions are interesting for all link analysis algorithms. 
In this paper, we will be dealing with PageRank.  
By answering these questions, we wish to gain insights into the connectivity of large, real-world graphs and the quality of documents, and provide a better ranking.
A result of this study is a ranking method that takes advantage of the most influential back links to discover authorities and communities.

\section{The Best Back Links and MaxRank}\label{sec:maxrank}
In the case of PageRank-style authority discovery, 
a natural definition of the {\em best back link} of a page is the one with the largest score.  

We discover the best back links in the same process of authority score update, giving a so-called {\em MaxRank} method.  
The basic idea of this algorithm is, 
with probability $\lambda$, 
the contributing score comes from the best backlink of the page;
with probability $1-\lambda$, 
the contributing scores come from a random backlink of the page.

\subsection{The Algorithm}
In particular, MaxRank of a page $j$ ($j =1,2, \ldots, N$) is defined by
\begin{align}
{R}(j) = & c \left[ \lambda P(i^*,j){ {R}(i^*)}  + (1-\lambda) \sum_{i\in \mathcal{B}(j)}{ P(i,j){R}(i)} \right] \label{MaxRgen2}   \\
+ &  (1-c)v(j), \nonumber 
\end{align}
where 
\[
i^*=\arg\max_{i\in \mathcal{B}(j)}{ {R}(i)},
\]
$\lambda \in [0,1]$,
$\mathcal{B}(j)$ is the set of backlink pages of page $j$,
and $P({i,j})$ is the probability of going from page $i$ to page $j$,

 \subsection{Convergence of MaxRank}
 In this section, we first give a theorem showing that both variants of MaxRank are well defined. 
 A straightforward application of this theorem is that power iteration of computing MaxRank is guranteed to converge for $\lambda \in [0,1]$.

\begin{thm}\label{thm:well}
 For $c\in(0,1)$ and $\lambda \in [0, 1]$, 
 MaxRank is well defined.
 \end{thm}
 
 \begin{proof}
 For notational convenience, we define
 \[
 \mathbb{T}(R,j) = \left[ \lambda P(i^*,j)\max_{i\in \mathcal{B}(j)}{ {R}(i)}  + (1-\lambda) \sum_{i\in \mathcal{B}(j)}{ P(i,j){R}(i)} \right]. 
 \]
 Accordingly, we have
 \[
 R(j) =c\mathbb{T}(R,j)  + (1-c)v(j)
 \]
 $j=1,2, \ldots, N$.
 In matrix form, we have
 \begin{equation}\label{RcT}
 R =c\mathbb{T}(R)  +(1-c)v,
 \end{equation}
 where $ \mathbb{T}(R)$ is a vector, with $\mathbb{T}(R)(j) = \mathbb{T}(R,j) $, $j=1, 2, \ldots, N$.
 
 Next we are to prove that $\mathbb{T}(R)$ is a non-expansion operator with respect to the $1$-norm, which means that 
 \begin{equation}\label{tmp1}
 ||\mathbb{T}(R) ||_{1} \le || R||_{1}.
 \end{equation}
 According to the definition of $\mathbb{T}(R)$, $\mathbb{T}(R)(j)$ and $\mathbb{T}(R,j)$, we have
 \[
 \mathbb{T}(R)= T \cdot R,
 \]
 where ${T}$ is a $N\times N$ matrix, with $T({j,i})=  P({i,j})$, if page $i$ is the best backlink of page $j$;
 otherwise ${T}({j,i})= (1-\lambda) P({i,j})$.
 
 Then the inequality (\ref{tmp1}) can be proven in the following steps:
 \begin{align*}
 ||\mathbb{T}(R) ||_{1} &\le   ||T||_{1}  ||R||_{1}\nonumber \\ 
 &=   \max_{i=1,2, \ldots, N} \sum_{j=1}^{N}{T}(j,i) ||R||_{1} \nonumber \\
 &\le  \max_{i=1,2, \ldots, N} \sum_{j=1}^{N} P(i,j) ||R||_{1} \nonumber \\
 &=   ||R||_{1} \nonumber 
 \end{align*}
 For the third equation, the equality holds when $i$ is the best backlink for all pages. 
 
 Thus $T$ is a non-expansion mapping in $1$-norm. 
 According to equation (\ref{RcT}), $R$ is defined by a contraction mapping composed of $T$ and $c$. Hence $R$ is finite. 
 \end{proof}
 
 The definition of MaxRank enables straightforward estimation using power iteration starting from any initial guess. 
 The convergence of power iteration is guaranteed following an argument similar to Theorem \ref{thm:well}.
 
  \begin{thm}[Convergence]\label{thm:conv}
 For $c\in(0,1)$ and $\lambda \in [0, 1]$, 
 power iteration of solving MaxRank converges to the true vector defined in (\ref{MaxRgen2}), irrespective of any initial vector. 
 \end{thm}

We will consider the random surfer in the remainder of this paper. 
That is, the probability of going from a (non-dangling) page $i$ to a page $j$ is $1/n_{i}$, where $n_{i}$ is the number of (forward) links on page $i$.
In this case, it is noticeable that when $\lambda=0$, 
the algorithm reduces to PageRank.
When $\lambda=1$, the algorithm only considers the ``best'' backlink it finds and
ignores the contribution from the others. 
However, in our experience, 
this usually gives poor ranking results because the selected best backlink pages are usually not good in quality. 
\section{Empirical Results}\label{sec:wiki}
In this section, we study the proposed algorithm and questions on the Wikipedia English article dump, 
which contains about $6$ million pages (articles or categories). 
For all algorithms, $c=0.85$ was used. 
The teleportation probabilities were uniformly set to $1/N$.
All algorithms are updated by the standard power iteration. 
No sophisticated update is used for any algorithm.

Recall that we would like to study the following questions. 
{\em What are the sources of the best back links? How many are they? How influential are they?} 
For space limitation we show only the case of $\lambda=0.1$ in this paper. 

\subsection{Sources of the Best Back links}
\label{subsec:exp_wiki_Q2}

Table \ref{tab:MaxRanked08501} shows the sources of the best backlinks for the top-50 pages on Wikipedia, 
using algorithm MaxRank{\em ed} with $\lambda=0.1$. 
Note that this choice produces a similar scoring to PageRank, as will be shown later.
The sources of the best backlinks are mostly global authorities. 
The very top pages are seen to support many top pages. 
For example, ``United States'' influences many other concepts which further influence the remaining of the site. 
The effect is that this classifies the site into clusters of nodes, in each of which there are only a small number of dominant nodes. 

 \begin{figure}[t]
 \centering
 \includegraphics[width=2.5in]{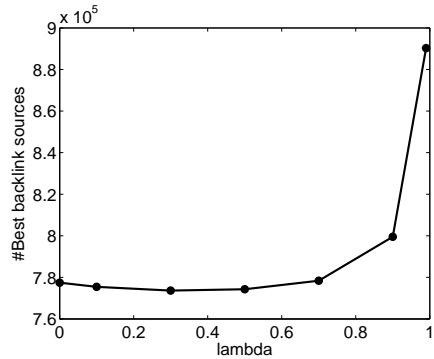}
 \caption{Number of best backlink sources on Wikipedia according to MaxRanked with $\lambda=0,0.1, 0.3, 0.5, 0.7, 0.9, 0.99$.
 }\label{fig:wiki:maxranked:lam:coresize}
 \end{figure}

There are only  $775,438$ unique backlink sources with MaxRank. 
They support the whole site and form a core. 
The size of this core is only about $0.7\%$ of the total number of links ($117,864,053$), and about $13.5\%$ of the total number of the pages ($5,743,047$). 
On average a core page ``supports'' about $3,620,343/775,438\approx 4.7$ pages. 
This is also an estimate of the average size of the clusters. 
The size of the best backlink core for various $\lambda$ is shown in Figure \ref{fig:wiki:maxranked:lam:coresize}.
$\lambda=0.3$ leads to the smallest core for this example. 
For $\lambda$ larger than $0.7$, the core size is much larger and increases much quicker with respect to $\lambda$.

 \begin{figure*}[t]
\begin{minipage}[t]{0.2\textwidth}
\centering
\vspace{0pt}
 \includegraphics[width=2.2in]{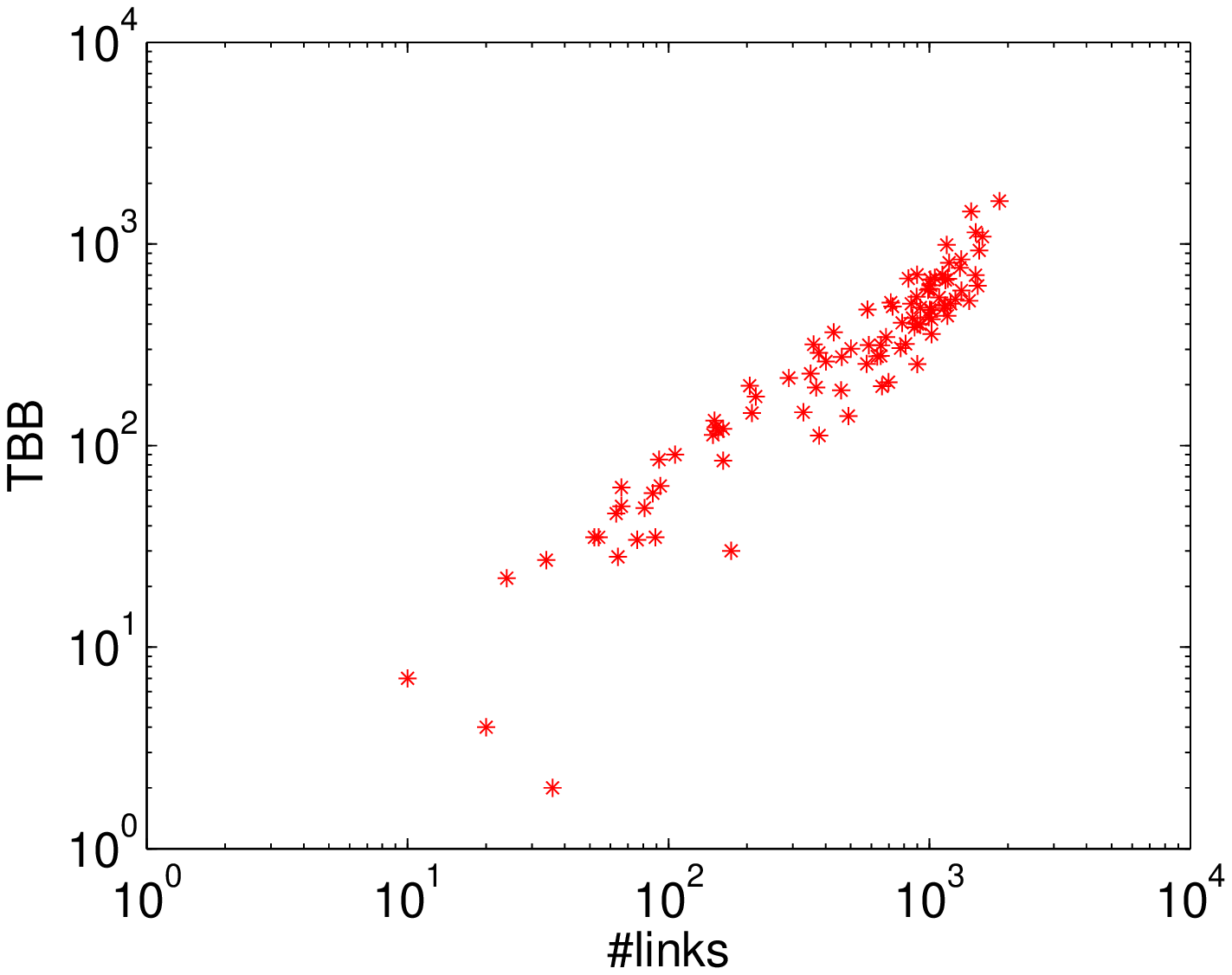}
\end{minipage}%
\hspace{1.8cm}
\begin{minipage}[t]{0.2\textwidth}
\vspace{0pt}
\centering
 \includegraphics[width=2.2in]{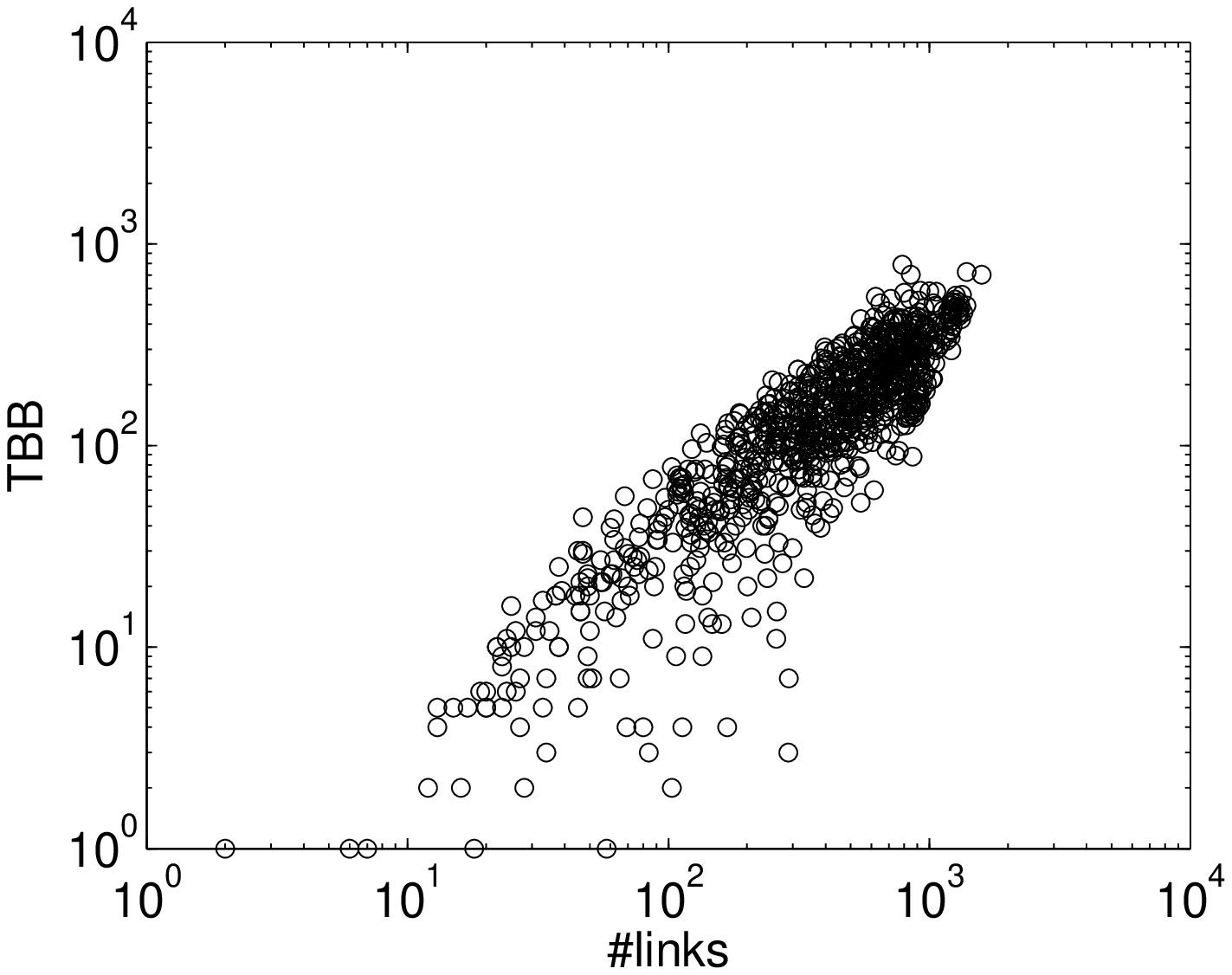}
\end{minipage}%
\hspace{1.8cm}
\begin{minipage}[t]{0.2\textwidth}
\vspace{0pt}
\centering
 \includegraphics[width=2.2in]{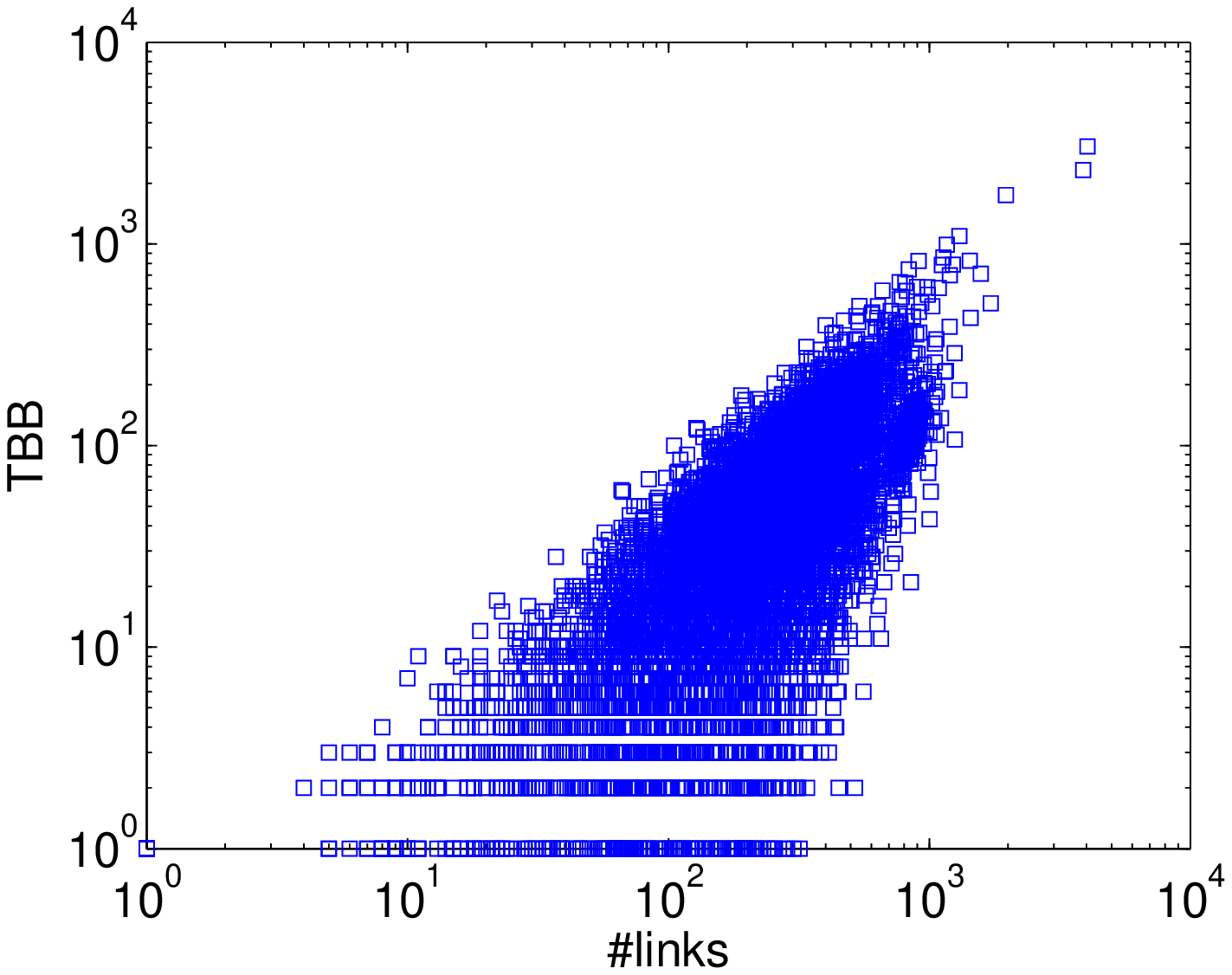}
\end{minipage}%
 \caption{Left to Right: The TBB of the top 1--100, 101--1000, 1001--10000 authorities on Wikipedia ($\lambda=0.1$).
 }\label{fig:wiki:maxranked:01:tbbtopauth}
 \end{figure*}

\subsection{Influence of the Best Backlink Sources}
We measure the influence of the best backlink sources in three distinct aspects. 
The first measure is the {\em collective influence} of the  best backlink sources in the graph, which is defined as 
the ratio of the sum of the scores of all the best backlink sources over the sum of the scores of all the pages. 
The {collective influence} of the core is $53.1\%$. Note that the number of core pages is only about $13.5\%$ of the whole graph. 
Thus the influence of the core is significant.

Different core sources have different strength of influence. 
Some contribute many best backlinks, while others only contribute a few. 
Thus this suggests a measure for influential sources, in particular, by the {\em times of being the best back link (TBB)} to other nodes. 
Note that the TBB of an influential page is equal to the number of pages that the page supports. 
Table \ref{tab:maxranked_times08501} shows the ordering of the sources according to the TBB measure. 
In addition, we also show in this table the ratio of TBB to the out-degree of the sources, which measures the percentage of competitive links cast by the sources. 
In this top list we see many hubs and authorities, and the number of hubs is more than the number of authorities. 
Thus on Wikipedia the more links an article has the more likely it is influential to others.

A log-log plot of the distributions of the out-degree and the TBB is shown in the left plot of Figure \ref{fig:wiki:maxranked:01:outdegrees_influencedratio_tbbnf}.
Some key observations are as follows.
First, the number of pages that have been the best backlink only a few times is very large, 
while the number of pages that have been the best backlink many times is very small. 
Second, the log-log curve of TBB distribution is more straight, which means the TBB distribution follows an exponential distribution in a more strict
way.
Third, 
it can be seen that for $x>10$, the two curves follow a similar exponential distribution with a close exponent, 
with a large shift in the $x$ direction which indicates that the TBB of a page is much smaller than the out-degree of the same page.

The (sorted) ratio between the TBB and the out-degree for each core page is shown in the middle plot of Figure \ref{fig:wiki:maxranked:01:outdegrees_influencedratio_tbbnf}.
First, the sources whose value of the ratio is smaller than $0.2$ are about $66\%$ of the total core. 
This means the majority of the core has only $20\%$ of their links being the best backlinks. 
Second, the number of those pages whose ratio is equal to $1.0$ is about $87,193$ ($11\%$ of the total core).
Astonishingly, $86,763$ ($99.5\%$) of them have only one link. 
This sheds lights on the structure of Wikipedia.
Most of them are due to the existence of ``redirect pages'' in Wikipedia, which contains no content but a ``link'' to another article.
Third, the remaining sources have a ratio larger than $0.2$.
Together with the nontrivial sources whose ratio is $1.0$, they form the most competitive link sources of the core.
They take about $23\%$ of the total core.
Of them, only $4,632$ sources have a ratio larger than 0.5,
and $360$ sources have a ratio larger than $0.8$.
In short,
{\em the number of nontrivial, competitive backlink sources is very small}.

The third is from the perspective of an ordinary page (either in the core or not in the core), a measure of being influenced by the best back link, by
the ratio of the score contributed by the best back link over the overall score of the page.  
We expect this measure can distinguish authorities.  
This ratio for all the pages is shown in the right plot of Figure \ref{fig:wiki:maxranked:01:outdegrees_influencedratio_tbbnf}. 
For authorities with high scores, this ratio is very small. 
Thus they are not easily influenced even by the best backlink source. 
As pages become less authoritative (along the negative direction of the $x$-axis), the values of this ratio become more diverse. 
For example, 
we can observe the values of this ratio cover almost the whole range of $(0, 1)$ for pages with a score equal to $10^{-5}$.

Figure \ref{fig:wiki:maxranked:01:tbbtopauth} shows the TBB versus the out-degree for the top authorities. 
For the very top-100 authorities, the curve is almost linear, and very close to $y=x$. 
(Note that all points are below $y=x$.)
Thus their links are very influential. 
Further down the ordering of the authorities, 
we can observe that there are more and more less influential pages.

 \begin{figure*}[t]
  \begin{minipage}[t]{0.2\textwidth}
  \vspace{0pt}
  \centering
\includegraphics[width=2.2in]{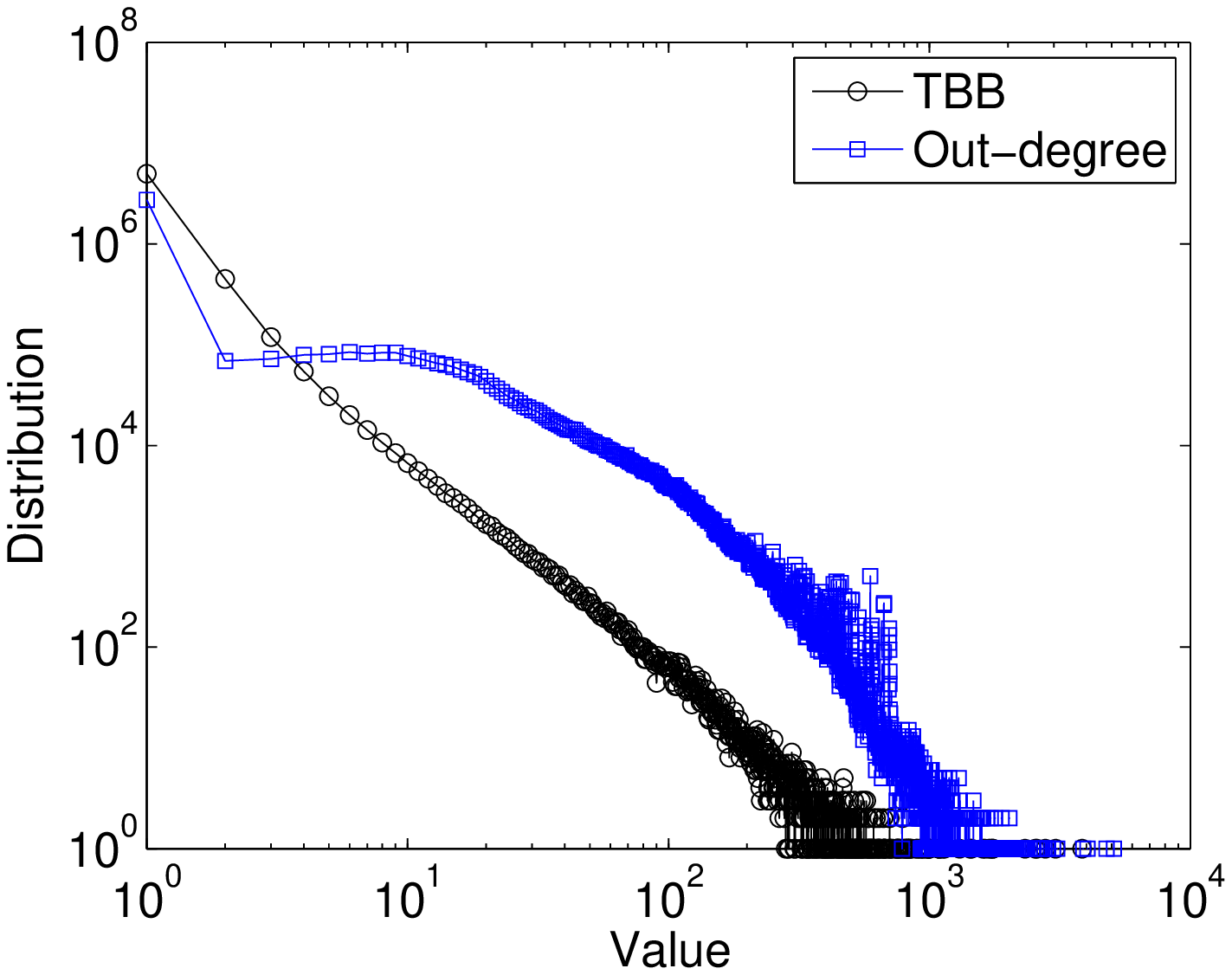}
  \end{minipage}%
   \hspace{1.8cm}
 \begin{minipage}[t]{0.2\textwidth}
 \vspace{0pt}
 \centering
 \includegraphics[width=2.2in]{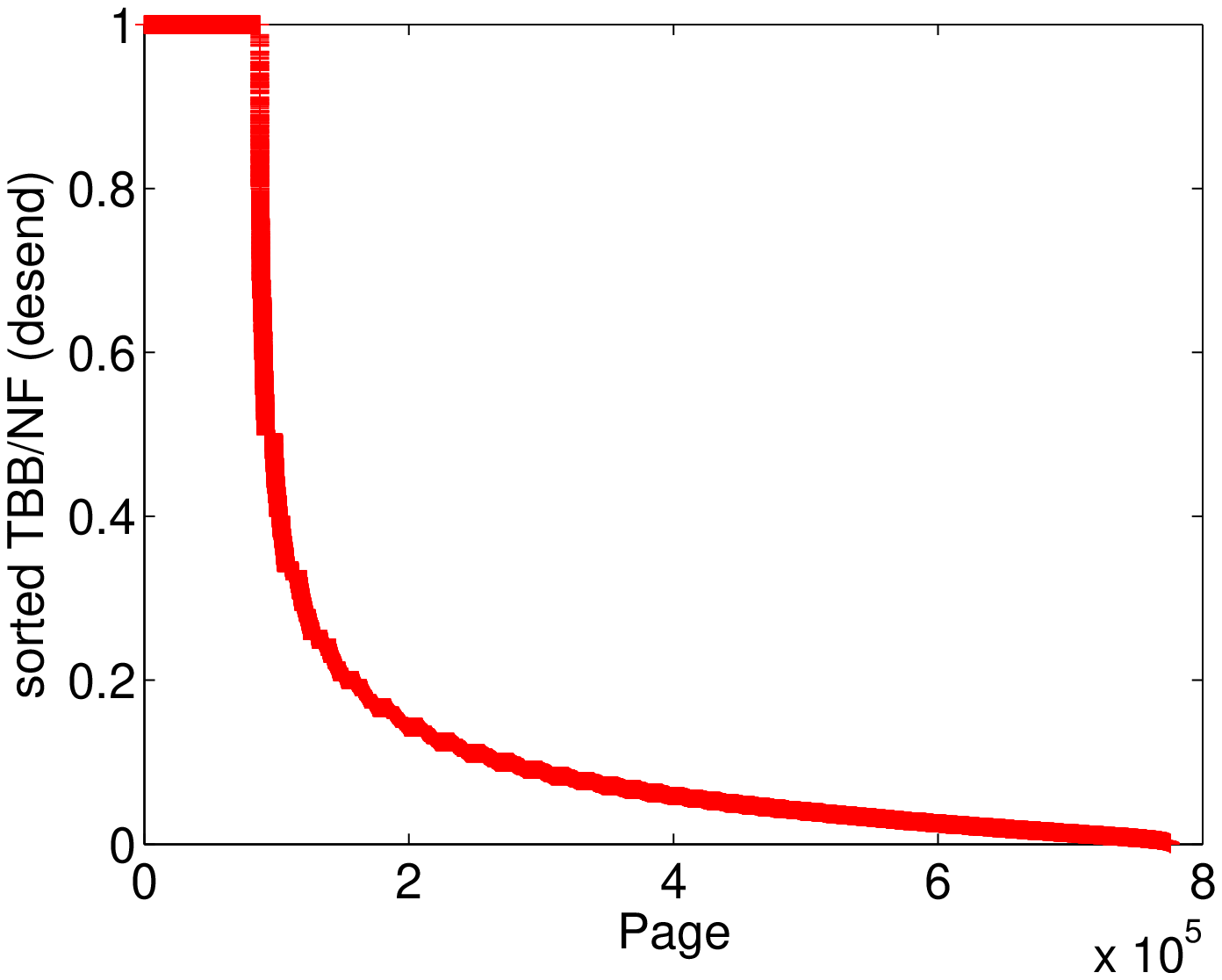}
 \end{minipage}%
\hspace{1.8cm}
 \begin{minipage}[t]{0.2\textwidth}
 \vspace{0pt}
 \centering
 \includegraphics[width=2.2in]{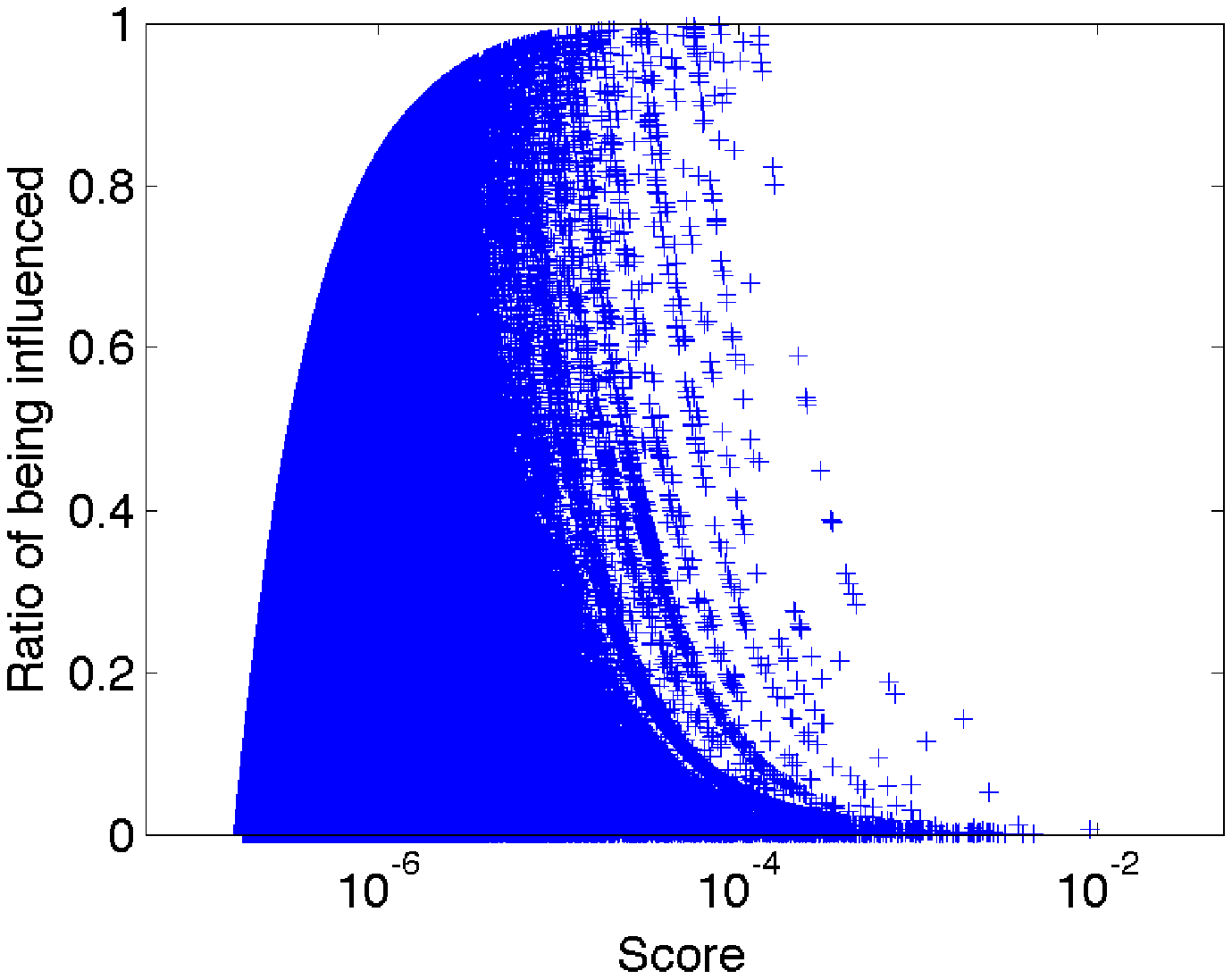}
 \end{minipage}%
  \caption{Left: Distributions of TBB and out-degree for the best backlink sources on Wikipedia.
  Middle: The sorted ratio between TBB and out-degree for the best backlink sources. 
  Right: The ratio of the score being influenced by the best backlink source for all pages (with a nonzero number of backlinks). $\lambda=0.1$.
 }\label{fig:wiki:maxranked:01:outdegrees_influencedratio_tbbnf}
 \end{figure*}


\begin{table*}[t]
\caption{The ordering of the best backlink sources according to the ``Times of being the Best Backlinks'' (TBB). $\lambda=0.1$.
}
\label{tab:maxranked_times08501}
\vskip 0.15in
\begin{center}
\begin{scriptsize}
\begin{tabular}{llllll}
\hline
Rank &Page &TBB &\#Links  & TBB/\#Links &Score\\
\hline
1 & List\_of\_endangered\_animal\_species & 3850 & 5097 & 0.755346 & 0.000003 \\
2 & Area\_codes\_in\_Germany & 3044 & 4032 & 0.754960 & 0.000084 \\
3 & Village\_Development\_Committee & 2784 & 2982 & 0.933602 & 0.000024 \\
4 & Index\_of\_India-related\_articles & 2581 & 4772 & 0.540863 & 0.000006 \\
5 & List\_of\_Tachinidae\_genera\_and\_species & 2547 & 2635 & 0.966603 & 0.000003 \\
6 & List\_of\_years & 2325 & 3885 & 0.598456 & 0.000105 \\
7 & List\_of\_auxiliaries\_of\_the\_United\_States\_Navy & 1751 & 1886 & 0.928420 & 0.000017 \\
8 & List\_of\_municipalities\_of\_Switzerland & 1749 & 1967 & 0.889171 & 0.000049 \\
9 & List\_of\_Bulbophyllum\_species & 1717 & 1742 & 0.985649 & 0.000001 \\
10 & List\_of\_municipalities\_and\_towns\_in\_Slovakia & 1692 & 2302 & 0.735013 & 0.000008 \\
11 & 2007 & 1632 & 1856 & 0.879310 & 0.001658 \\
12 & List\_of\_state\_leaders\_by\_year & 1604 & 2026 & 0.791708 & 0.000007 \\
13 & United\_States & 1448 & 1448 & 1.000000 & 0.009093 \\
14 & List\_of\_Roman\_Catholic\_dioceses\_(alphabetical) & 1434 & 2326 & 0.616509 & 0.000002 \\
15 & List\_of\_cutaneous\_conditions & 1306 & 1829 & 0.714051 & 0.000013 \\
16 & United\_Kingdom & 1140 & 1505 & 0.757475 & 0.003863 \\
17 & List\_of\_Olympic\_medalists\_in\_athletics\_(men) & 1127 & 2069 & 0.544708 & 0.000013 \\
18 & List\_of\_postal\_codes\_in\_Germany & 1092 & 1304 & 0.837423 & 0.000084 \\
19 & List\_of\_Vanity\_Fair\_caricatures & 1090 & 1815 & 0.600551 & 0.000001 \\
20 & Russia & 1087 & 1597 & 0.680651 & 0.001464 \\
21 & List\_of\_mantis\_genera\_and\_species & 1001 & 1057 & 0.947020 & 0.000025 \\
22 & List\_of\_United\_States\_Representatives\_from\_New\_York & 996 & 1342 & 0.742176 & 0.000025 \\
23 & List\_of\_extant\_baronetcies & 992 & 1168 & 0.849315 & 0.000033 \\
24 & Catholic\_Church & 990 & 1166 & 0.849057 & 0.000913 \\
25 & Index\_of\_statistics\_articles & 984 & 2026 & 0.485686 & 0.000016 \\
26 & 2006\_in\_music & 953 & 2162 & 0.440796 & 0.000017 \\
27 & List\_of\_prehistoric\_bony\_fish & 946 & 1049 & 0.901811 & 0.000007 \\
28 & England & 927 & 1551 & 0.597679 & 0.002462 \\
29 & List\_of\_EC\_numbers\_(EC\_2) & 920 & 1141 & 0.806310 & 0.000006 \\
30 & List\_of\_marine\_aquarium\_fish\_species & 898 & 1103 & 0.814143 & 0.000002 \\
31 & List\_of\_school\_districts\_in\_Texas & 891 & 943 & 0.944857 & 0.000007 \\
32 & Peerage\_of\_the\_United\_Kingdom & 856 & 1132 & 0.756184 & 0.000052 \\
33 & List\_of\_rivers\_of\_New\_Zealand & 843 & 919 & 0.917301 & 0.000019 \\
34 & List\_of\_chess\_players & 841 & 1577 & 0.533291 & 0.000003 \\
35 & London & 838 & 1323 & 0.633409 & 0.001363 \\
36 & List\_of\_subjects\_in\_Gray\'s\_Anatomy:\_IX.\_Neurology & 825 & 1426 & 0.578541 & 0.000037 \\
37 & 2004\_in\_music & 824 & 1780 & 0.462921 & 0.000018 \\
38 & Pronunciation\_of\_asteroid\_names & 823 & 910 & 0.904396 & 0.000030 \\
39 & List\_of\_destroyers\_of\_the\_United\_States\_Navy & 814 & 1030 & 0.790291 & 0.000012 \\
40 & California & 807 & 1192 & 0.677013 & 0.001055 \\
41 & List\_of\_EC\_numbers\_(EC\_1) & 799 & 1063 & 0.751646 & 0.000013 \\
42 & Cocaine & 789 & 1231 & 0.640942 & 0.000059 \\
43 & Sibley-Monroe\_checklist\_18 & 789 & 1844 & 0.427874 & 0.000000 \\
44 & List\_of\_United\_States\_Representatives\_from\_Pennsylvania & 789 & 1085 & 0.727189 & 0.000021 \\
45 & List\_of\_Digimon & 788 & 788 & 1.000000 & 0.000186 \\
46 & List\_of\_bird\_genera & 788 & 1929 & 0.408502 & 0.000005 \\
47 & List\_of\_subjects\_in\_Gray\'s\_Anatomy:\_XI.\_Splanchnology & 786 & 1116 & 0.704301 & 0.000046 \\
48 & List\_of\_State\_Routes\_in\_New\_York & 772 & 985 & 0.783756 & 0.000025 \\
49 & Italy & 762 & 1310 & 0.581679 & 0.001633 \\
50 & List\_of\_ICF\_Canoe\_Sprint\_World\_Championships\_medalists\_in\_men\'s\_kayak & 761 & 865 & 0.879769 & 0.000007 \\

\hline
\end{tabular}
\end{scriptsize}
\end{center}
\vskip -0.1in
\end{table*}

 \subsection{Convergence Studies}

Figure \ref{fig:prec_tau_maxranked} (Left) compares the convergence rates of MaxRank and PageRank, measured in terms of the 
(1-norm) errors between successive iterations. 
MaxRank is faster than PageRank.
The advantage is very significant for large $\lambda$.
MaxRank with $\lambda=0.1$ needs about $20$ iterations to reach the accuracy by PageRank at the $30$th iteration, 
while with $\lambda=0.9$ MaxRank only needs $3$ or $4$ iterations.

\subsection{Performance of MaxRank}
We compared the top list for the three algorithms, since it is usually the most important in practice.
Table \ref{tab:MaxRanked0850} shows the top $50$ pages by PageRank (MaxRank with $\lambda=0$).
Table \ref{tab:MaxRanked08501}, Table \ref{tab:MaxRanked08505} and Table \ref{tab:MaxRanked08509} show the top results of MaxRank with $\lambda=0.1, 0.5, 0.9$.
We also tested $\lambda=1$ for MaxRank, but the results were very poor. 
The intuition is that the found ``best backlinks'' are not good without considering the wisdom of the majority.
Note that in the tables,
``ISBN is short for ``International\_Standard\_Book\_Number'', and 
``Inter-Air-Trans-code'' is short for ``International\_Air\_Transport\_Association\_airport\_code''.


The top lists of these algorithms have some similarities, and also some differences.  
In order to measure the similarity between the algorithms, 
we performed comparisons using two measurements. 
One is the percentage of common pages in the top-$k$ lists by two algorithms, 
\[
c_{k} = \frac{\mbox{\# Common pages in top-$k$}}{ k} \in [0,1].
\] 
The other is Kendall's tau coefficient which measures the correlation in two rankings \citep{kendall}.
Here we care about whether MaxRank ranks the top-$k$ pages of PageRank in a consistent manner to PageRank, 
so the measure used is
\[
\tau_{k} = \frac{n_{k} }{ C_{k}^{2}} \in [0,1],
\]
where $n_{k}$ is the number of concordant orderings for every two pages from the top-$k$ pages of PageRank.
The results of $c_{k}$
are summarized in the middle plot of Figure \ref{fig:prec_tau_maxranked}, for $k=$5, 10, 30, 50, 80, 100, 300, 500, 800, 1000.
Notice that MaxRank performs remarkably similarly to PageRank for $\lambda=0.1$, due to that the effect of the best backlinks is made small.
In general, the smaller $\lambda$ is, the more similar ranking of MaxRank to that of PageRank.

The results of $\tau_{k}$ are summarized in the right plot of Figure \ref{fig:prec_tau_maxranked}.
Similarly,
the smaller the parameter $\lambda$ is, 
the more similar the ranking is to PageRank. 
In particular, MaxRank with $\lambda=0.1$ has a very similar $\tau_{k}$ to PageRank for all $k$.  
For large $\lambda$ like $0.9$ and $0.99$, 
MaxRank still has about $80\%$ similarities on average, 
and $65\%$ similarities at worst to PageRank.
The difference between the orderings of $0.9$ and $0.99$ for MaxRank is relatively small for all $k$.
This suggests that increasing $\lambda$ to large values close to $1$ produces stable rankings. 

%

\begin{figure*}[t]
\begin{minipage}[t]{0.2\textwidth}
\vspace{0pt}
\centering
\includegraphics[width=2.3in]{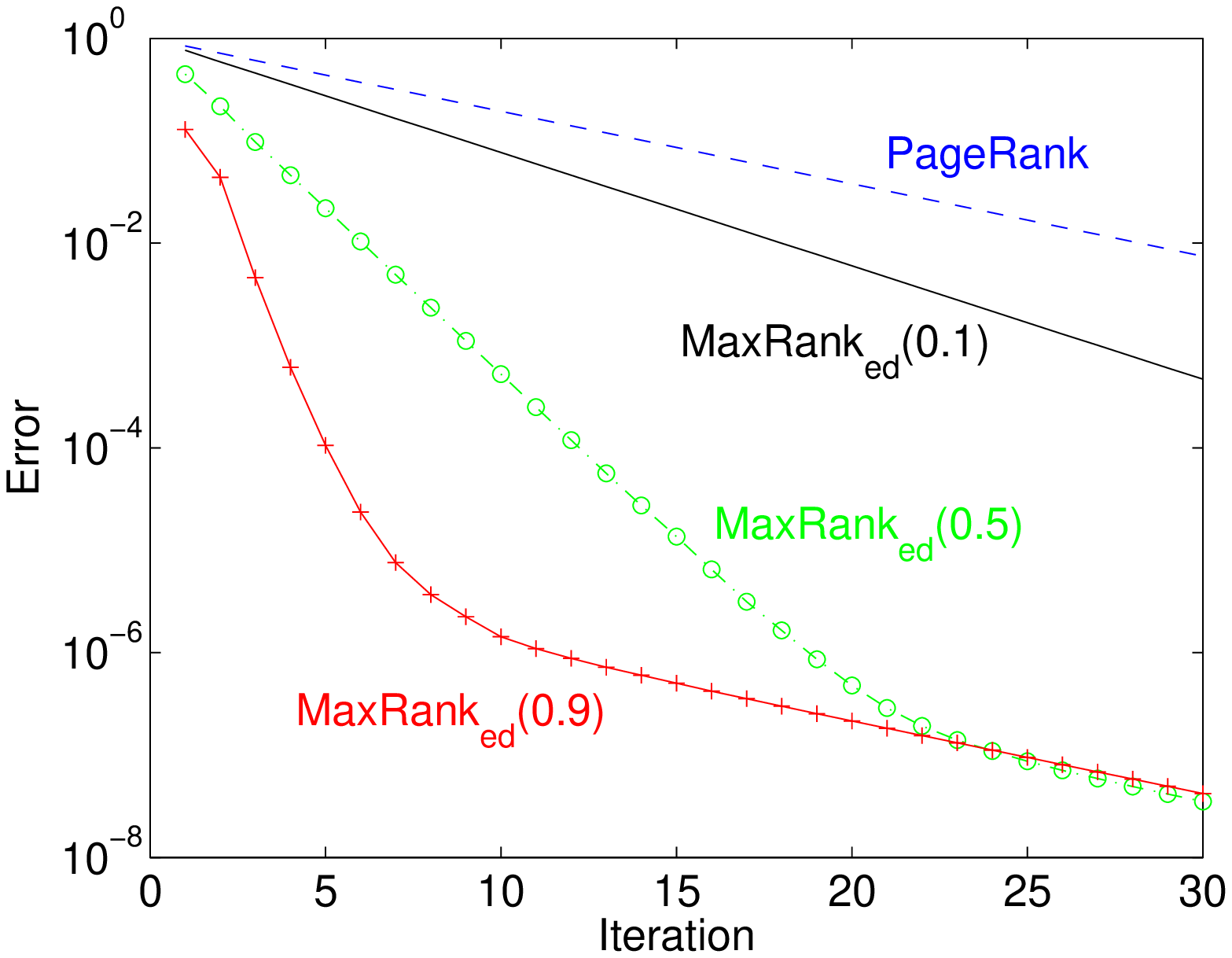}
\end{minipage}%
\hspace{1.8cm}
\begin{minipage}[t]{0.2\textwidth}
\vspace{0pt}
\centering
\includegraphics[width=2.2in]{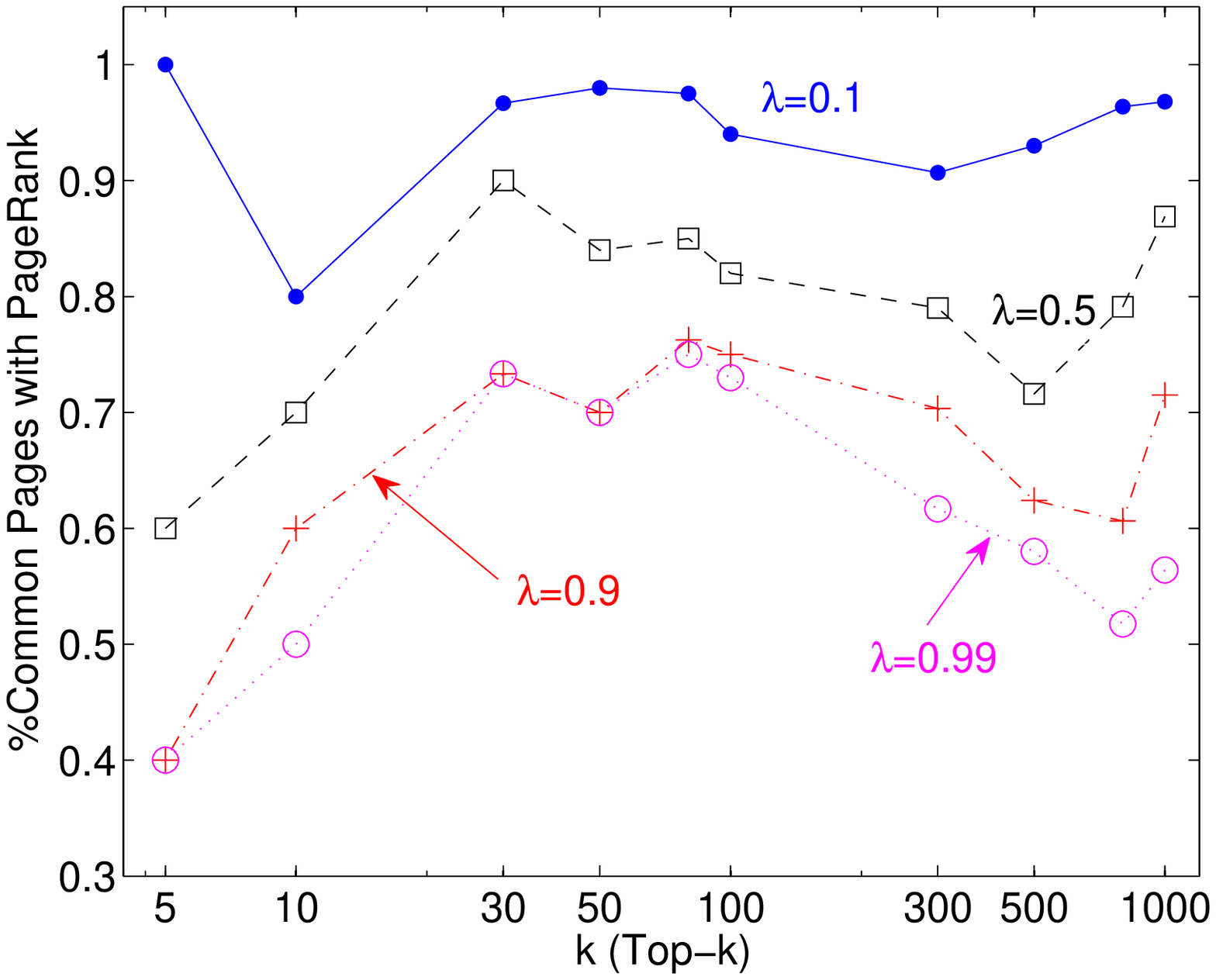}
\end{minipage}%
\hspace{1.8cm}
\begin{minipage}[t]{0.2\textwidth}
\vspace{0pt}
\centering
\includegraphics[width=2.3in]{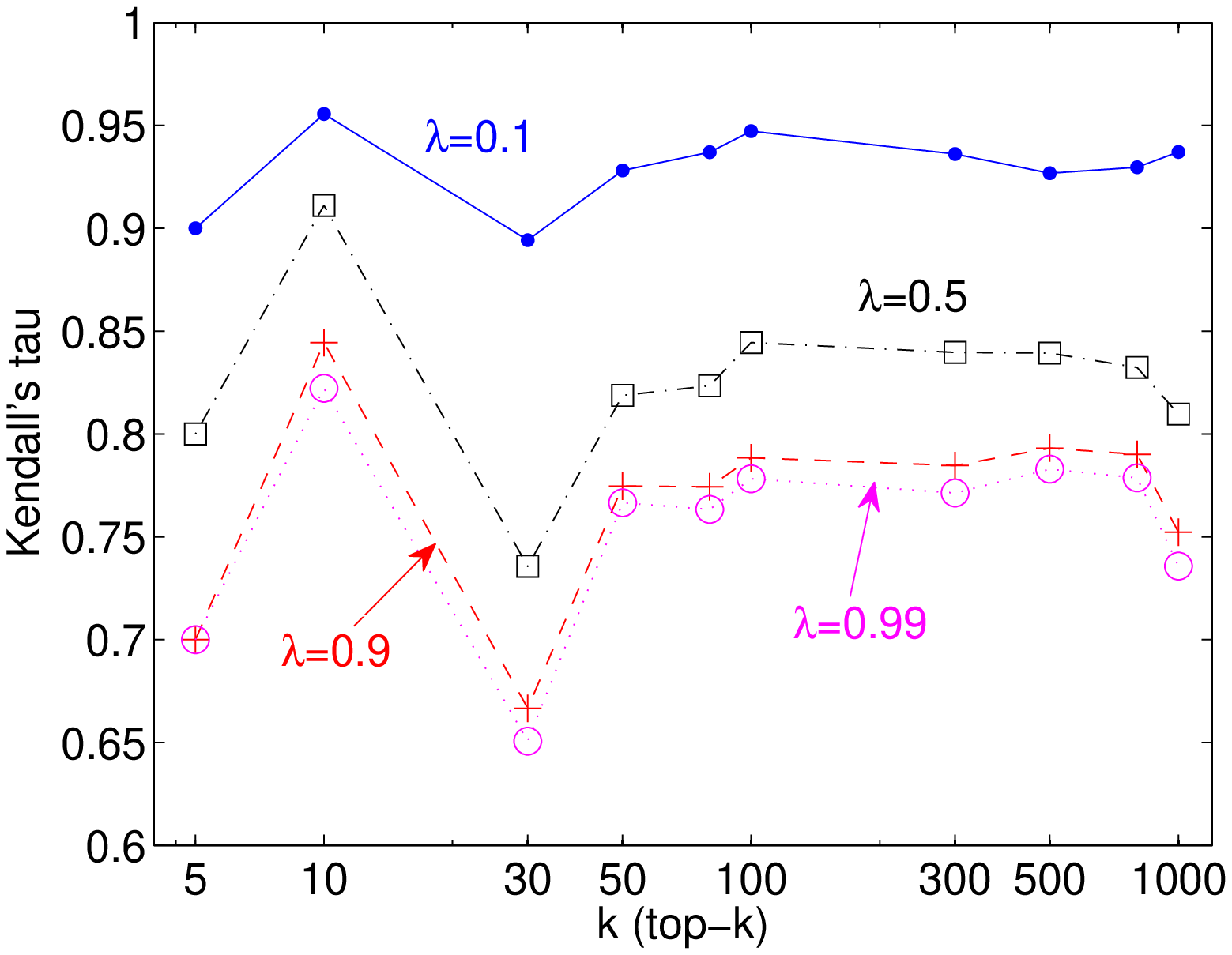}
\end{minipage}%
\caption{
Left: convergence rate comparisons of MaxRank and PageRank. 
Middle and Right: Percentage of the common pages and Kendall's tau for the top-$k$ lists of MaxRank and PageRank. 
}
\label{fig:prec_tau_maxranked}
\end{figure*}

\begin{table}[t]
\caption{{\small Top $50$ Wikipedia pages by MaxRank($0$) (PageRank)}. 
}
\label{tab:MaxRanked0850}
\vskip 0.15in
\begin{center}
\begin{scriptsize}
\begin{tabular}{llll}
\hline
Rank&Page & Score &Best backlink\\
\hline
1 & United\_States & 0.013911 & ``ISBN'' \\
2 & ``ISBN'' & 0.007283 & United\_States \\
3 & United\_Kingdom & 0.006135 & United\_States \\
4 & Wikimedia\_Commons & 0.005986 & Wiktionary \\
5 & Wiktionary & 0.004151 & Wikimedia\_Commons \\
6 & France & 0.004081 & United\_States \\
7 & Canada & 0.004049 & United\_States \\
8 & Biography & 0.003964 & Wiki \\
9 & Germany & 0.003860 & United\_States \\
10 & England & 0.003766 & United\_States \\
11 & Biological\_classification & 0.003562 & Arthropod \\
12 & English\_language & 0.003522 & United\_States \\
13 & Australia & 0.003426 & United\_States \\
14 & World\_War\_II & 0.003186 & United\_States \\
15 & Binomial\_nomenclature & 0.003176 & Biological\_classification \\
16 & Japan & 0.003091 & United\_States \\
17 & India & 0.003026 & United\_States \\
18 & Internet\_Movie\_Database & 0.002907 & Alexa\_Internet \\
19 & Abbreviation & 0.002882 & USA \\
20 & Music\_genre & 0.002844 & Poland \\
21 & Association\_football & 0.002766 & United\_States \\
22 & Europe & 0.002751 & United\_States \\
23 & Record\_label & 0.002734 & Music\_genre \\
24 & Italy & 0.002612 & United\_States \\
25 & 2007 & 0.002505 & Australia \\
26 & Russia & 0.002339 & United\_States \\
27 & London & 0.002152 & United\_Kingdom \\
28 & Spain & 0.002150 & United\_States \\
29 & Latin & 0.002077 & United\_States \\
30 & 2006 & 0.002030 & Germany \\
31 & Personal\_name & 0.001989 & Given\_name \\
32 & 2008 & 0.001919 & Germany \\
33 & New\_York\_City & 0.001850 & United\_States \\
34 & Netherlands & 0.001841 & United\_States \\
35 & Poland & 0.001829 & United\_States \\
36 & Sweden & 0.001825 & United\_States \\
37 & Scientific\_name & 0.001752 & Biological\_classification \\
38 & Public\_domain & 0.001752 & Wikimedia\_Commons \\
39 & Brazil & 0.001663 & United\_States \\
40 & Time\_zone & 0.001663 & United\_States \\
41 & China & 0.001658 & World\_War\_II \\
42 & French\_language & 0.001651 & United\_States \\
43 & World\_War\_I & 0.001643 & United\_States \\
44 & Catholic\_Church & 0.001623 & France \\
45 & California & 0.001620 & United\_States \\
46 & New\_Zealand & 0.001592 & United\_States \\
47 & Area & 0.001569 & United\_States \\
48 & 2005 & 0.001559 & France \\
49 & New\_York & 0.001554 & United\_States \\
50 & German\_language & 0.001515 & United\_States \\
\hline
\end{tabular}
\end{scriptsize}
\end{center}
\vskip -0.1in
\end{table}

\begin{table}[t]
\caption{Top $50$ Wikipedia pages by MaxRank, $\lambda=0.1$.
}
\label{tab:MaxRanked08501}
\vskip 0.15in
\begin{center}
\begin{scriptsize}
\begin{tabular}{llll}
\hline
Rank&Page & Score &Best backlink\\
\hline
1 & United\_States & 0.009093 & ``ISBN'' \\
2 & ``ISBN''& 0.004404 & United\_States \\
3 & United\_Kingdom & 0.003863 & United\_States \\
4 & Wikimedia\_Commons & 0.003614 & Wiktionary \\
5 & Biography & 0.003035 & Wiki \\
6 & Biological\_classification & 0.002773 & Arthropod \\
7 & Canada & 0.002626 & United\_States \\
8 & France & 0.002546 & United\_States \\
9 & Wiktionary & 0.002474 & Wikimedia\_Commons \\
10 & England & 0.002462 & United\_States \\
11 & Germany & 0.002452 & United\_States \\
12 & Binomial\_nomenclature & 0.002273 & Biological\_classification \\
13 & Australia & 0.002203 & United\_States \\
14 & English\_language & 0.002172 & United\_States \\
15 & Music\_genre & 0.002130 & Poland \\
16 & Record\_label & 0.002054 & Music\_genre \\
17 & Internet\_Movie\_Database & 0.002036 & Alexa\_Internet \\
18 & Japan & 0.002021 & United\_States \\
19 & India & 0.001968 & United\_States \\
20 & World\_War\_II & 0.001946 & United\_States \\
21 & Association\_football & 0.001910 & United\_States \\
22 & Abbreviation & 0.001786 & USA \\
23 & Europe & 0.001690 & United\_States \\
24 & 2007 & 0.001658 & Australia \\
25 & Italy & 0.001633 & United\_States \\
26 & Personal\_name & 0.001559 & Given\_name \\
27 & Russia & 0.001464 & United\_States \\
28 & London & 0.001363 & United\_Kingdom \\
29 & Spain & 0.001339 & United\_States \\
30 & 2006 & 0.001322 & Germany \\
31 & 2008 & 0.001274 & Germany \\
32 & Scientific\_name & 0.001236 & Biological\_classification \\
33 & Poland & 0.001206 & United\_States \\
34 & New\_York\_City & 0.001183 & United\_States \\
35 & Sweden & 0.001155 & United\_States \\
36 & Latin & 0.001140 & United\_States \\
37 & Netherlands & 0.001136 & United\_States \\
38 & Public\_domain & 0.001120 & Wikimedia\_Commons \\
39 & Time\_zone & 0.001077 & United\_States \\
40 & Brazil & 0.001076 & United\_States \\
41 & California & 0.001055 & United\_States \\
42 & Record\_producer & 0.001024 & Music\_genre \\
43 & China & 0.001023 & Japan \\
44 & New\_Zealand & 0.001007 & United\_States \\
45 & 2005 & 0.001006 & France \\
46 & World\_War\_I & 0.001004 & United\_States \\
47 & New\_York & 0.000999 & United\_States \\
48 & Romania & 0.000968 & United\_States \\
49 & Area & 0.000966 & United\_States \\
50 & Politician & 0.000964 & Video\_game \\
\hline
\end{tabular}
\end{scriptsize}
\end{center}
\vskip -0.1in
\end{table}

\begin{table}[t]
\caption{Top $50$ Wikipedia pages by MaxRank, $\lambda=0.5$.
}
\label{tab:MaxRanked08505}
\vskip 0.15in
\begin{center}
\begin{scriptsize}
\begin{tabular}{llll}
\hline
Rank&Page & Score &Best backlink\\
\hline
1 & United\_States & 0.002444 & ``ISBN'' \\
2 & Biography & 0.001198 & Genre \\
3 & Biological\_classification & 0.001103 & Arthropod \\
4 & ``ISBN'' & 0.000949 & United\_States \\
5 & United\_Kingdom & 0.000929 & United\_States \\
6 & Music\_genre & 0.000753 & Record\_producer \\
7 & Record\_label & 0.000726 & Music\_genre \\
8 & Wikimedia\_Commons & 0.000712 & Association\_football \\
9 & Canada & 0.000685 & United\_States \\
10 & England & 0.000671 & United\_States \\
11 & Binomial\_nomenclature & 0.000656 & Biological\_classification \\
12 & Personal\_name & 0.000640 & Given\_name \\
13 & Internet\_Movie\_Database & 0.000630 & Royal\_Navy \\
14 & Germany & 0.000607 & United\_States \\
15 & France & 0.000603 & United\_States \\
16 & Australia & 0.000558 & United\_States \\
17 & India & 0.000543 & United\_States \\
18 & Association\_football & 0.000533 & United\_States \\
19 & Japan & 0.000523 & United\_States \\
20 & Wiktionary & 0.000510 & Wikimedia\_Commons \\
21 & English\_language & 0.000502 & United\_States \\
22 & 2007 & 0.000454 & Australia \\
23 & Abbreviation & 0.000436 & USA \\
24 & Arthropod & 0.000423 & Lepidoptera \\
25 & World\_War\_II & 0.000415 & United\_States \\
26 & Italy & 0.000392 & United\_States \\
27 & Europe & 0.000373 & United\_States \\
28 & Studio\_album & 0.000370 & Music\_genre \\
29 & Politician & 0.000364 & Video\_game \\
30 & Record\_producer & 0.000363 & Music\_genre \\
31 & 2008 & 0.000349 & Germany \\
32 & Russia & 0.000343 & United\_States \\
33 & 2006 & 0.000340 & Germany \\
34 & London & 0.000332 & United\_Kingdom \\
35 & Scientific\_name & 0.000326 & Biological\_classification \\
36 & Poland & 0.000318 & United\_States \\
37 & Spain & 0.000313 & United\_States \\
38 & Romania & 0.000310 & United\_States \\
39 & New\_York\_City & 0.000294 & United\_States \\
40 & Public\_domain & 0.000293 & Wikimedia\_Commons \\
41 & Time\_zone & 0.000288 & United\_States \\
42 & Brazil & 0.000282 & United\_States \\
43 & Sweden & 0.000280 & United\_States \\
44 & California & 0.000278 & United\_States \\
45 & Television & 0.000270 & United\_Kingdom \\
46 & Drainage\_basin & 0.000267 & New\_York\_City \\
47 & Netherlands & 0.000259 & United\_States \\
48 & 2005 & 0.000255 & France \\
49 & New\_York & 0.000249 & United\_States \\
50 & New\_Zealand & 0.000249 & United\_States \\
\hline
\end{tabular}
\end{scriptsize}
\end{center}
\vskip -0.1in
\end{table}

\begin{table}[t]
\caption{
Top $50$ Wikipedia pages by MaxRank, $\lambda=0.9$.
}
\label{tab:MaxRanked08509}
\vskip 0.15in
\begin{center}
\begin{scriptsize}
\begin{tabular}{lllll}
\hline
Rank&Page & Score &Best backlink\\
\hline
1 & United\_States & 0.000311 & United\_Kingdom \\
2 & Biography & 0.000194 & Autobiography \\
3 & Biological\_classification & 0.000175 & Arthropod \\
4 & Music\_genre & 0.000111 & Record\_producer \\
5 & United\_Kingdom & 0.000109 & United\_States \\
6 & Record\_label & 0.000106 & Music\_genre \\
7 & Personal\_name & 0.000105 & Given\_name \\
8 & ``ISBN'' & 0.000099 & United\_States \\
9 & England & 0.000088 & United\_States \\
10 & Internet\_Movie\_Database & 0.000087 & Royal\_Navy \\
11 & Canada & 0.000085 & United\_States \\
12 & Binomial\_nomenclature & 0.000079 & Biological\_classification \\
13 & Arthropod & 0.000075 & Lepidoptera \\
14 & India & 0.000073 & United\_States \\
15 & Germany & 0.000072 & United\_States \\
16 & France & 0.000069 & United\_States \\
17 & Australia & 0.000067 & United\_States \\
18 & Association\_football & 0.000063 & United\_States \\
19 & Japan & 0.000062 & United\_States \\
20 & Wikimedia\_Commons & 0.000061 & Arthropod \\
21 & Politician & 0.000059 & Video\_game \\
22 & Studio\_album & 0.000058 & Music\_genre \\
23 & Abbreviation & 0.000058 & USA \\
24 & English\_language & 0.000057 & United\_States \\
25 & 2007 & 0.000057 & Australia \\
26 & Record\_producer & 0.000054 & Music\_genre \\
27 & Wiktionary & 0.000052 & Wikimedia\_Commons \\
28 & Geocode & 0.000048 & UN/LOCODE \\
29 & UN/LOCODE & 0.000047 & ``Inter-Air-Trans-code'' \\
30 & Italy & 0.000046 & United\_States \\
31 & 2008 & 0.000044 & Germany \\
32 & Romania & 0.000044 & United\_States \\
33 & Lepidoptera & 0.000042 & Moth \\
34 & World\_War\_II & 0.000042 & United\_States \\
35 & 2006 & 0.000040 & Germany \\
36 & Drainage\_basin & 0.000039 & New\_York\_City \\
37 & London & 0.000039 & United\_Kingdom \\
38 & Television & 0.000039 & United\_Kingdom \\
39 & Europe & 0.000039 & United\_States \\
40 & Poland & 0.000038 & United\_States \\
41 & Time\_zone & 0.000038 & United\_States \\
42 & Russia & 0.000038 & United\_States \\
43 & Genus & 0.000037 & Biological\_classification \\
44 & Conservation\_status & 0.000036 & IUCN\_Red\_List \\
45 & Spain & 0.000036 & United\_States \\
46 & Public\_domain & 0.000035 & Wikimedia\_Commons \\
47 & Brazil & 0.000035 & United\_States \\
48 & New\_York\_City & 0.000035 & United\_States \\
49 & California & 0.000035 & United\_States \\
50 & IUCN\_Red\_List & 0.000034 & Conservation\_status \\
\hline
\end{tabular}
\end{scriptsize}
\end{center}
\vskip -0.1in
\end{table}

\section{Discussion}\label{sec:discussion}

The size of the Web creates a large computation burden for PageRank.
Currently most large commercial search engines index $10$ to $100$ billion pages. 
However, the Web is actually much larger, e.g.,
there were already 1 trillion unique URLs in 2008 according to Google. 
\footnote{http://googleblog.blogspot.com/2008/07/we-knew-the-Web-was-big.html}
Computing a single, global PageRank for the Web is already very demanding. 
Page et. al. used power iteration to take advantage of the sparse nature of the link structure of the Web \citep{Page98}.
Kamvar et al. proposed an adaptive method which monitors the change in the PageRank update for each page, and removes those pages whose update no longer changes \citep{TaherAda}. 
Methods proposed in \citep{Tahereff}, \citep{pr_block}, and \citep{PR_survey} take advantage of the structure of link matrices and compute PageRank block-wise.
Other linear system solvers were also used to update PageRank.
For example, 
Kamvar et. al. used extrapolation methods \citep{PR_extra};
Gleich et. al. proposed an inner-outer iteration procedure \citep{gleichthesis,PR_inout}, which essentially applies preconditioning incrementally; and
Langville and Meyer, and Ipsen and Kirklad studied aggregation/disaggregation methods \citep{amy06,iad_conv}.

It becomes more severe when one wants to computes many score vectors, 
such as many personalized PageRank \citep{Page98,jeh03}, context-sensive or query-dependent scores \citep{davood,qdpr,taher02}, which has numerous applications in search systems. 
Jeh and Widom proposed a scalable method by pre-computing some components of PageRank and saving them for efficient future computation \citep{jeh03}. 
Fogaras et. al. simulated a number of random walks and used Monte-Carlo methods to estimate personalized PageRank vectors \citep{ppr_fully}.

The computation of PageRank is very demanding. 
Thus distributed, parallel computation becomes necessary for large graphs.
The methods in \citep{graphagg04,amy_agg04,pr_agg} partition the whole graph into disjoint subgraphs, 
and then compute local PageRank for each subgraph. 
The local PageRanks are then merged, considering the links between the subgraphs.
The methods in \citep{pr_lin,pr_parallel} feature in the use of advanced linear system solvers.
Some researchers also considered efficient hardware structures, such as specially optimized circuits \citep{pr_fpga}.  
For excellent surveys of PageRank, please refer to \citep{PR_survey,pr_survey3,revisit06,amy06,liu07,pr_survey2}.

Our method is very different from these efforts in literature, though
it should be noted that these techniques also apply to our algorithms in a straightforward way. 
Our work takes advantage of the influential links in updating PageRank-style scores. 
We hope by doing so one can gain speedup in convergence and the performance is similar or comparable to PageRank.

\section{Conclusion}\label{sec:con}
The observation leading to this paper is that there exists one or more valuable backlinks for a page with a nonzero number of backlinks. 
We show that by leveraging the best backlinks a recursive update 
can have a much faster convergence than PageRank.
The algorithm has a parameter $\lambda\in[0,1]$, which controls the effects of the best backlinks discovered. 
When $\lambda=0$, the algorithm reduces to PageRank. 
Empirical results show that with large $\lambda$s (but smaller than $1$) the new algorithm converges dramatically faster,
but the results still have $80\%$ similarities to PageRank on average (measured with Kendall's tau).
Thus our algorithm is advantageous for ranking in large search systems, where the computation of many personalized, query-dependent or context-sensitive score vectors is demanding.

Results on Wikipedia show that the number of unique best backlink sources (the so-called ``core'' in the paper) is only about $13.5\%$ of the total number of pages. 
However, the sum of their scores is more than a half (about $53.1\%$) of the total scores. 
We propose to measure a source in the core by the times of being the best backlinks (TBB) and the ratio between TBB and the out-degree.
Results show that TBB follows an exponential distribution with a similar exponent to the distribution of the out-degrees. 
With these two measures,
the number of competitive backlink sources is very small. 
Results also show that a top authority is not easily influenced by the best backlink source.

 \begin{scriptsize}
 
\end{scriptsize}

\pagenumbering{arabic}
\setcounter{page}{1}


\begin{thebibliography}{}

\bibitem[Arasu et~al., 2001]{searchweb}
Arasu, A., Cho, J., Garcia-Molina, H., Paepcke, A., and Raghavan, S. (2001).
\newblock Searching the {Web}.
\newblock {\em ACM Transactions on Internet Technology}, 1(1):2--43.

\bibitem[Berkhin, 2005]{PR_survey}
Berkhin, P. (2005).
\newblock A survey on {PageRank} computing.
\newblock {\em Internet Mathematics}, 2(1):73--120.

\bibitem[Bianchini et~al., 2005]{pr_survey3}
Bianchini, M., Gori, M., and Scarselli, F. (2005).
\newblock Inside pagerank.
\newblock {\em ACM Transactions on Internet Technologies}, 5(1):92--128.

\bibitem[Brinkmeier, 2006]{revisit06}
Brinkmeier, M. (2006).
\newblock {PageRank} revisited.
\newblock {\em ACM Transactions on Internet Technology}, 6(3):282--301.

\bibitem[Broder et~al., 2004]{graphagg04}
Broder, A.~Z., Lempel, R., Maghoul, F., and Pedersen, J. (2004).
\newblock Efficient {PageRank} approximation via graph aggregation.
\newblock {\em WWW}.

\bibitem[Fogaras and R{\'a}cz, 2004]{ppr_fully}
Fogaras, D. and R{\'a}cz, B. (2004).
\newblock Towards scaling fully personalized {PageRank}.
\newblock {\em WAW}.

\bibitem[Gleich and Zhukov, 2005]{pr_parallel}
Gleich, D. and Zhukov, L. (2005).
\newblock Scalable computing for power law graphs: Experience with parallel
  pagerank.
\newblock Technical report, Yahoo! Research Labs Technical Report.

\bibitem[Gleich et~al., 2004]{pr_lin}
Gleich, D., Zhukov, L., and Berkhin, P. (2004).
\newblock Fast parallel {PageRank}: A linear system approach.
\newblock Technical report, Yahoo! Research Labs Technical Report,
  YRL-2004-038.

\bibitem[Gleich, 2009]{gleichthesis}
Gleich, D.~F. (2009).
\newblock {\em Models and Algorithms for PageRank sensitivity}.
\newblock PhD thesis, Stanford University.

\bibitem[Gleich et~al., 2010]{PR_inout}
Gleich, D.~F., Gray, A.~P., Greif, C., and Lau, T. (2010).
\newblock An inner-outer iteration for {PageRank}.
\newblock {\em SIAM Journal of Scientific Computing}, 32(1):349--371.

\bibitem[Haveliwala, 1999]{Tahereff}
Haveliwala, T. (1999).
\newblock Effcient computation of {PageRank}.
\newblock Technical Report 1999--31, Database Group, Computer Science
  Department, Stanford University.

\bibitem[Haveliwala, 2005]{Taherthesis}
Haveliwala, T. (2005).
\newblock {\em Context-Sensitive Web Search}.
\newblock PhD thesis, Stanford University.

\bibitem[Haveliwala, 2002]{taher02}
Haveliwala, T.~H. (2002).
\newblock Topic-sensitive {PageRank}.
\newblock {\em WWW}.

\bibitem[Ipsen and Kirklad, 2004]{iad_conv}
Ipsen, C.~F. and Kirklad, S. (2004).
\newblock Convergence analysis of an improved pagerank algorithm.
\newblock Technical report, NCSU CRSC Technical Report.

\bibitem[Jeh and Widom, 2003]{jeh03}
Jeh, G. and Widom, J. (2003).
\newblock Scaling personalized web search.
\newblock {\em WWW}.

\bibitem[Kamvar et~al., 2003a]{TaherAda}
Kamvar, S., Haveliwala, T., and Golub, G. (2003a).
\newblock Adaptive methods for the computation of {PageRank}.
\newblock Technical report, Stanford University.

\bibitem[Kamvar et~al., 2003b]{pr_block}
Kamvar, S., Haveliwala, T., Manning, C., and Golub., G. (2003b).
\newblock Exploiting the block structure of the web for computing {{PageRank}}.
\newblock Technical report, Stanford University.

\bibitem[Kamvar et~al., 2003c]{PR_extra}
Kamvar, S.~D., Haveliwala, T.~H., Manning, C.~D., and Golub, G.~H. (2003c).
\newblock Extrapolation methods for accelerating {PageRank} computations.
\newblock {\em WWW}.

\bibitem[Kendall, 1975]{kendall}
Kendall, S. M.~G. (1975).
\newblock {\em Rank Correlation Methods}.
\newblock Charles Griffin \& Company Limited.

\bibitem[Kleinberg, 1998]{HITS}
Kleinberg, J. (1998).
\newblock Authoritative sources in a hyperlinked environment.
\newblock {\em SODA}.

\bibitem[Langville and Meyer, 2004a]{amy_agg04}
Langville, A. and Meyer, C. (2004a).
\newblock Updating {{PageRank}} with iterative aggregation.
\newblock {\em WWW}.

\bibitem[Langville and Meyer, 2004b]{deeper04}
Langville, A.~N. and Meyer, C.~D. (2004b).
\newblock Deeper inside {PageRank}.
\newblock {\em Internet Mathematics}, 1.

\bibitem[Langville and Meyer, 2006]{amy06}
Langville, A.~N. and Meyer, C.~D. (2006).
\newblock {\em Google's {PageRank} and Beyond: The Science of Search Engine
  Rankings}.
\newblock Princeton University Press.

\bibitem[Liu, 2007]{liu07}
Liu, B. (2007).
\newblock {\em Web Data Mining: Exploring Hyperlinks, Contents and Usage data}.
\newblock Springer.

\bibitem[Page et~al., 1998]{Page98}
Page, L., Brin, S., Motwani, R., and Winograd, T. (1998).
\newblock The {PageRank} citation ranking: Bringing order to the web.
\newblock Technical report, Stanford University.

\bibitem[Rafiei and Mendelzon, 2000]{davood}
Rafiei, D. and Mendelzon, A.~O. (2000).
\newblock What is this page known for? computing web page reputations.
\newblock {\em WWW}.

\bibitem[Richardson and Domingos, 2002]{qdpr}
Richardson, M. and Domingos, P. (2002).
\newblock The intelligent surfer: Probabilistic combination of link and content
  information in {PageRank}.
\newblock {\em NIPS}.

\bibitem[Sargolzaei and Soleymani, 2010]{pr_survey2}
Sargolzaei, P. and Soleymani, F. (2010).
\newblock Pagerank problem, survey and future research directions.
\newblock {\em International Mathematical Forum}, 5(19):937--956.

\bibitem[Seamas et~al., 2008]{pr_fpga}
Seamas, M., Dermot, G., and McElroy, C. (2008).
\newblock Towards an {FPGA} solver for the pagerank eigenvector problem.
\newblock {\em Advances in Parallel Computing}, 15.

\bibitem[Zhu et~al., 2005]{pr_agg}
Zhu, Y., Ye, S., and Li, X. (2005).
\newblock Distributed {PageRank} computation based on iterative
  aggregation-disaggregation methods.
\newblock {\em CIKM}.

\end{thebibliography}

\end{document}